\documentclass[letterpaper,10pt,conference]{ieeeconf}
 \IEEEoverridecommandlockouts
\makeatletter
\let\NAT@parse\undefined
\makeatother

\usepackage[square, sort&compress, numbers]{natbib}

\usepackage{amsmath,graphicx,epsfig,color,amsfonts, algorithm, subfigure}

\usepackage{balance}

\graphicspath{{./fig/}}

\usepackage{version,xspace}
\usepackage[noend]{algorithmic}

\def\real{\mathbb{R}}

\newcommand{\until}[1]{\{1,\dots, #1\}}
\newcommand{\subscr}[2]{#1_{\textup{#2}}}

\newcommand{\setdef}[2]{\{#1 \; | \; #2\}}

\newcommand{\map}[3]{#1: #2 \rightarrow #3}

\newcommand{\argmin}{\operatorname{argmin}}

\newcommand\oprocendsymbol{\hbox{$\square$}}
\newcommand\oprocend{\relax\ifmmode\else\unskip\hfill\fi\oprocendsymbol}

\renewcommand{\baselinestretch}{1}

\renewcommand{\theenumi}{\roman{enumi}}

\newcommand\bit[1]{\textit{\textbf{#1}}}

\def \bs {\boldsymbol}
\def \mc {\mathcal}

\newtheorem{theorem}{Theorem}
\newtheorem{proposition}{Proposition}
\newtheorem{lemma}{Lemma}
\newtheorem{corollary}{Corollary}

\newtheorem{remark}{Remark}

\newtheorem{assumption}{Assumption}

\renewcommand{\theenumi}{(\roman{enumi}}

\title{SIS Epidemic Model under Mobility on Multi-layer Networks
\thanks{This work was supported by ARO grant W911NF-18-1-0325.}
}
\author{Vishal Abhishek$^{1}$ and Vaibhav Srivastava$^{2}$
\thanks{$^{1}$Vishal Abhishek is with the Department of Mechanical Engineering,
       Michigan State University,
       East Lansing, MI 48824-1226, USA
        {\tt\small abhishe3 at egr.msu.edu}}%
\thanks{$^{2}$Vaibhav Srivastava is with the Electrical and Computer Engineering, Michigan State University, East Lansing, MI 48824-1226, USA
        {\tt\small vaibhav at egr.msu.edu}}%
}

\begin{document}

\maketitle

\begin{abstract}
We study the influence of heterogeneous mobility patterns in a population on the SIS epidemic model. In particular, we consider a patchy environment in which each patch comprises individuals belonging the different classes, e.g., individuals in different socio-economic strata. We model the mobility of individuals of each class across different patches through an associated Continuous Time Markov Chain (CTMC). The topology of these multiple CTMCs constitute the multi-layer network of mobility. At each time, individuals move in  the multi-layer network of spatially-distributed patches according  to  their CTMC  and subsequently  interact  with  the  local individuals in the patch  according to  an  SIS  epidemic model. We derive a deterministic continuum limit model describing these mobility-epidemic interactions. We establish the existence of a Disease-Free Equilibrium (DFE) and an Endemic Equilibrium (EE)  under  different  parameter regimes and establish their (almost) global asymptotic stability using  Lyapunov  techniques.  We derive simple sufficient conditions that highlight the influence of the multi-layer network on the stability of DFE.  
Finally, we numerically illustrate that the derived model provides a good approximation to the stochastic model with a finite population and also demonstrate the influence of the multi-layer network structure  on  the  transient  performance.
\end{abstract}

\section{Introduction}

Contagion dynamics are used to model a variety of phenomena such as spread of influence, disease and rumors. 
Epidemic propagation models are  a class of contagion models that have been used in the context of disease spread~\cite{anderson1992infectious, DE-JK:10}, spread of computer viruses~\cite{kleinberg2007computing, wang2009understanding}, routing in mobile communication networks~\cite{zhang2007performance}, and spread of rumors~\cite{jin2013epidemiological}. Epidemic propagation in patchy environments refers to the epidemic spread process in an environment comprised of disjoint  spatially distributed regions (patches). In these models, individuals interact within each patch and also move across different patches according to a CTMC. 


In this paper, we consider a generalized epidemic propagation model in a patchy environment in which individuals within each patch belong to multiple classes, and individuals within each class move according to an associated CTMC. This leads to a multi-layer mobility model and we study its interaction with epidemic propagation. 
Using Lyapunov techniques, we characterize the steady state behavior of the model under different parameter regimes and characterize the influence of mobility on epidemic dynamics. 



Epidemic models have been extensively studied in the literature. The two most widely studied models are SIS (Susceptible-Infected-Susceptible) and SIR (Susceptible-Infected-Recovered)  models, wherein individuals are classified into one of the three categories: susceptible, infected or recovered. In contrast to the classical SIS/SIR models where the dynamics of the fraction of the population in each category~\cite{DE-JK:10} is studied, the networked models consider patches clustered into different nodes, and the patch-level dynamics is determined by the local SIS/SIR interactions as well as the interactions with neighboring patches in the network graph~\cite{AnalysisandControlofEpidemics_ControlSysMagazine_Pappas, meiBullo2017ReviewPaper_DeterministicEpidemicNetworks, fall2007epidemiological, khanafer_Basar2016stabilityEpidemicDirectedGraph}. While most of studies on SIR/SIS epidemic models focus on continuous-time dynamics, the authors in~\cite{Hassibi2013globaldynEpidemics,Ruhi_Hassibi2015SIRS} study network epidemic dynamics in discrete time setting. 

Some common generalizations of the SIR/SIS models include: SEIR model~\cite{ AnalysisandControlofEpidemics_ControlSysMagazine_Pappas, mesbahi2010graph}, where an additional classification ``exposed" is introduced, SIRS~\cite{DE-JK:10,Ruhi_Hassibi2015SIRS}, where individuals get temporary immunity after recovery and then become susceptible again, and SIRI~\cite{gomez2015abruptTransitionsSIRI, pagliara_NaomiL2018bistability,pagliara2019adaptive}, where after recovery, agents become susceptible with a different rate of infection. The network epidemic dynamics have also been studied for time-varying networks~ \cite{bokharaie2010_EpidemicTVnetwork, Preciado2016_EpidemicTVnetwork,  Beck2018_EpidemicTimeVaryingNetwork}.


The terms population dispersal and network mobility have been used interchangeably in the literature. Epidemic spread under mobility has been modeled and analyzed as reaction-diffusion process in \cite{colizza2008epidemicReaction-DiffusionMetapopuln, Saldana2008continoustime_Reaction-DiffnMetapopln}.
Epidemic spread in a patchy environment with population dispersal has been modeled and studied in  \cite{wang2004epidemic, jin2005effect, li2009global}. In these works, the mobility or dispersal patterns depend on the state (susceptible or infected) of the individuals, and conditions for global stability of the disease-free equilibrium and an endemic equilibrium are derived. When the mobility patterns are identical for all individuals, then these models reduce to a model similar to the single-layer version of the multi-layer model studied in this paper. 


Epidemic spread with mobility on a multiplex network of patches has been modeled and studied in \cite{soriano2018spreading_MultiplexMobilityNetwork_Metapopln}. Authors of this work consider a discrete-time model in which, at each time, individuals randomly move to another node, participate in epidemic propagation and then return to their home node. 
A multi-species SEIR epidemic model with population dispersal has been analyzed in \cite{arino2005multi} and conditions for the global stability of a disease-free equilibrium are derived. Stability results for endemic equilibrium for the single species case are also derived. The population dispersal model in \cite{arino2005multi} is identical to the multi-layer mobility model studied in this paper. In contrast to \cite{arino2005multi}, we focus on SIS epidemic model and completely characterize the properties of the model, including existence, uniqueness, and stability of both disease-free and endemic equilibria.


In this paper, we study a coupled epidemic-mobility model comprised of a set of patches located in a multi-layer network of  spatially distributed regions. Individuals within each patch (region) can travel across regions according to a Continuous Time Markov Chain (CTMC) characterising their mobility pattern and upon reaching a new region participate in the local SIS epidemic process. We extend the results for the deterministic network SIS model~\cite{fall2007epidemiological, Hassibi2013globaldynEpidemics, khanafer_Basar2016stabilityEpidemicDirectedGraph, meiBullo2017ReviewPaper_DeterministicEpidemicNetworks} to the proposed model and characterize its steady state and stability properties.  

The major contributions of this paper are fourfold. First, we derive a deterministic continuum limit model describing the interaction of the SIS dynamics with the multi-layer Markovian mobility dynamics.  The obtained model is similar to the model studied in \cite{arino2005multi}; however, our presentation derives the model from first principles. 
%
%
Second, we rigorously characterize the existence and stability of the equilibrium points of the derived model under different parameter regimes. 
Third, for the stability of the disease-free equilibrium, 
we determine some useful sufficient conditions which highlight the influence of multi-layer mobility on the steady state behavior of the dynamics. 
Fourth, we numerically illustrate that the derived model is a good approximation to the stochastic model with a finite population. We also illustrate the influence of the network topology on the transient properties of the model. 

    
    
    




The remainder of this paper is organized in the following way. In Section \ref{Sec: Mobility Modeling as Continous-Time Markov Process}, we derive the epidemic model under multi-layer mobility as a continuum limit to two interacting stochastic processes. In Section \ref{sec: analysis}, we characterize the existence and stability of disease-free and endemic equilibrium for the derived model. In Section \ref{Sec: numerical studies}, we illustrate our results using numerical examples. Finally, we conclude in Section \ref{Sec: conclusions}. 
\medskip

\noindent
{\it Mathematical notation:}
For any two real vectors $\bs x$, $\bs y \in \real^n$, we denote:\\
$\bs x \gg \bs y$, if $x_i > y_i$ for all $i \in \until{n}$,\\
$\bs x \geq \bs y$, if $x_i \geq y_i$ for all $i \in \until{n}$,\\
$\bs x > \bs y$, if $x_i \geq y_i$ for all $i \in \until{n}$ and $\bs x \neq \bs y$.\\
For a square matrix $G$, radial abscissa $\map{\mu}{\real^{n\times n}}{\real}$ is defined by 
\[
\mu(G) = \max \setdef{\mathrm{Re}(\lambda)}{\lambda \text{ is an eigenvalue of $G$}}, 
\]
where $\mathrm{Re}(\cdot)$ denotes the real part of the argument. 
Spectral radius $\rho$ is defined by
\[
\rho(G) = \max \setdef{|\lambda|}{\lambda \text{ is an eigenvalue of $G$}}, 
\]
where $|(\cdot)|$ denotes the absolute value of the argument.
For any vector $\bs x = [x_1,\dots,x_n]^\top$, $X=\operatorname{diag}(\bs x)$ is a diagonal matrix with $X_{ii}=x_i$ for all $i \in \until{n}$.

\section{SIS Model under Multi-layer Markovian Mobility} \label{Sec: Mobility Modeling as Continous-Time Markov Process}

We consider $n$ sub-population of individuals that are located in distinct spatial regions (patches). We assume the individuals within each patch can be classified into two categories: (i) susceptible, and (ii) infected. We assume that the individuals within each patch are further grouped into $m$ classes which decide how they travel to other patches. Let the connectivity of these patches corresponding to the mobility pattern of each class $\alpha \in \until{m}$ be modeled by a digraph $\mc G^\alpha = (\mc V, \mc E^\alpha)$, where $\mc V =\until{n}$ is the node (patch) set and $\mc E^\alpha \subset \mc V \times \mc V$ is the edge set. We model the mobility of individuals on each graph $\mc G^\alpha$ using a Continuous Time Markov Chain (CTMC) with generator matrix $Q^\alpha$, whose $(i,j)$-th entry is $q^\alpha_{ij}$. The entry $q^\alpha_{i j} \ge 0$, $i \ne j$, is the instantaneous transition rate from node $i$ to node $j$, and $-q^\alpha_{ii}= \nu^\alpha_{i}$ is the total rate of transition out of node $i$, i.e., $\nu^\alpha_{i} = \sum_{j \ne i}q^\alpha_{i j}$. Here, $q^\alpha_{ij} >0$, if $(i,j) \in \mc E^\alpha$; and $q^\alpha_{ij}=0$, otherwise. Let $x^\alpha_{i}(t)$ be the number of individuals of class $\alpha$ in patch $i$ at time $t$. Let $p^\alpha_i \in [0,1]$ (respectively, $1-p^\alpha_i$) be the fraction of infected (respectively, susceptible) sub-population of class $\alpha$ at patch $i$. Define $\bs p^\alpha := [p^\alpha_{1},\dots,p^\alpha_n]^\top$, $\bs x^\alpha := [x^\alpha_{1},\dots,x^\alpha_n]^\top$, $\bs p := [(\bs p^1)^\top,\dots, (\bs p^m)^\top]^\top$ and $\bs x := [ (\bs x^1)^\top,\dots, (\bs x^m)^\top]^\top$.

We model the interaction of mobility with the epidemic process as follows. At each time $t$, individuals of each class $\alpha$ within each node move on graph $\mc G^\alpha$ according to the CTMC with generator matrix $Q^\alpha$ and then interact with individuals within their current node according to an SIS epidemic process.  For the epidemic process at node $i$, let $\beta_i >0$ and $\delta_i \ge 0$ be the infection and recovery rate, respectively. We let $B^\alpha >0$ and $D^\alpha \ge 0$ be the positive and non-negative diagonal matrices with entries $\beta_i$ and $\delta_i$, $i \in \until{n}$, respectively. Let $B$ and $D$ be the positive and non-negative diagonal matrices with block-diagonal entries $B^\alpha$ and $D^\alpha$, $\alpha \in \until{m}$, respectively. Let $P^\alpha:=\operatorname{diag}(\bs p^\alpha)$ and $P:=\operatorname{diag}(\bs p)$. We now derive the continuous time dynamics that captures the interaction of mobility and the SIS epidemic dynamics. 




\begin{proposition}[\bit{SIS model under mobility}]\label{prop:model}
The dynamics of the fractions of the infected sub-population $\bs p$ and the number of individuals $\bs x^\alpha$ under multi-layer Markovian mobility model with generator matrices $Q^\alpha$, and infection and recovery matrices $B$ and $D$, respectively, are
\begin{subequations} \label{eq_Model}
\begin{align}
    \dot{\bs{p}} & = (BF(\bs{x})-D-L(\bs{x}))\bs{p} - P BF(\bs{x}) \bs{p} \label{eq_p}\\ 
    \dot{\bs{x}}^\alpha & = (Q^\alpha)^\top \bs{x}^\alpha, \label{eq_x}
\end{align}
\end{subequations}
where $L$ is an $nm \times nm$ block-diagonal matrix with block-diagonal terms $L^\alpha$, $\alpha \in \until{m}$, $L^\alpha(\bs{x})$ is a matrix with entries 
\[
l^\alpha_{ij}(\bs x) = \begin{cases} \sum_{j\neq i}q^\alpha_{j i} \frac{x^\alpha_{j}}{x^\alpha_{i}}, & \text{if } i = j, \\
-q^\alpha_{j i} \frac{x^\alpha_{j}}{x^\alpha_{i}}, & \text{otherwise},
\end{cases}
\]
 $F(\bs x) := [\bar{F}^\top(\bs x),\dots,\bar{F}^\top(\bs x)]^\top$ be a row-concatenated $nm \times nm$ matrix with each $n \times nm$ block-row as $\bar{F}(\bs x):=[F^1 (\bs x),\dots,F^m (\bs x)]$, and  $F^\alpha$ as a diagonal matrix with entries $f^\alpha_i (\bs x):=\frac{x^\alpha_i}{\sum_\alpha x^\alpha_i}$, i.e., the fraction of total population at node $i$ contributed by class $\alpha$.
\end{proposition}
\medskip

\begin{proof}
Consider a small time increment $h>0$ at time $t$. Then the number of individuals of class $\alpha$ present at node $i$ after the evolution of CTMC in time-interval $[t, t+h)$ is
\begin{equation} \label{eq_popln}
   x^\alpha_{i}(t+h)= x^\alpha_{i}(t)(1-\nu^\alpha_{i}h)+ \displaystyle\sum_{j\neq i}q^\alpha_{j i} x^\alpha_{j}(t)h + o(h) .
\end{equation}
After the mobility, individuals within each node interact according to SIS dynamics. Thus, the fraction of infected population present at node $i$ is: 
\begin{align} \label{eq_SIS_No of Infected}
 & x^\alpha_{i}(t+h) p^\alpha_{i}(t+h) \nonumber\\ 
 &=- x^\alpha_{i}(t) \delta_{i} p^\alpha_{i}(t)h + x^\alpha_{i}(t)\beta_{i} \bar p_{i}(t)(1-p^\alpha_{i}(t))h \nonumber\\
  &\quad + x^\alpha_{i}(t) p^\alpha_{i}(t)(1-\nu^\alpha_{i}h) + \displaystyle\sum_{j\neq i}q^\alpha_{j i} p^\alpha_{j} x^\alpha_{j}(t)h + o(h). 
\end{align}
where $\bar p_i$ is the fraction of infected population at node $i$ and is given as:
\[
 \bar p_i:= \sum_\alpha f^\alpha_i p^\alpha_i.
\]

The first two terms on the right side of \eqref{eq_SIS_No of Infected} correspond to epidemic process within each node, whereas the last two terms correspond to infected individuals coming from other nodes due to mobility. Using the expression of $x^\alpha_{i}$ from \eqref{eq_popln} in \eqref{eq_SIS_No of Infected} and taking the limit $h \to 0^+$ gives
\begin{multline} \label{eq_SIS_pi}
   \dot{p}^\alpha_{i}= - \delta_{i} p^\alpha_{i} + \beta_{i} \bar p_{i}(1-p^\alpha_{i})
   -l^\alpha_{ii} p^\alpha_{i} - \displaystyle\sum_{j\neq i} l^\alpha_{i j} p^\alpha_{j} .
\end{multline}
Writing above in vector form gives:
\begin{equation}
    \dot{\bs{p}}^\alpha = (-D^\alpha-L^\alpha(\bs{x}^\alpha))\bs{p}^\alpha + B^\alpha \bar F(\bs{x}) \bs p- P^\alpha B^\alpha \bar F(\bs{x}) \bs{p} .\label{eq_p_alpha}
\end{equation}
Similarly taking limits in \eqref{eq_popln} yields
\begin{equation} \label{eq_popln_det}
   \dot{x}^\alpha_{i} = -\nu^\alpha_{i} x^\alpha_{i} + \displaystyle\sum_{j\neq i}q^\alpha_{j i} x^\alpha_{j}.
\end{equation}
Rewriting \eqref{eq_SIS_pi} and \eqref{eq_popln_det} in vector form establishes the proposition. 
\end{proof}


\section{Analysis of SIS Model under Multi-layer Markovian Mobility} \label{sec: analysis}
In this section, we analyze the SIS model under multi-layer mobility~\eqref{eq_Model} under the following standard assumption: 
\begin{assumption} \label{Assumption:StrongConnectivity}
 Digraph $\mc G^\alpha$ is strongly connected, for all $\alpha \in \until{m}$, which is equivalent to matrices $Q^\alpha$ being irreducible \cite{Bullo-book_Networks}. \oprocend
\end{assumption}
Let $\bs v^\alpha$ be the right eigenvector of $(Q^\alpha)^\top$ associated with eigenvalue at $0$. We assume that $\bs v^\alpha$ is scaled such that its inner product with the associated left eigenvector $\bs 1_{n}$ is unity, i.e., $\bs 1_{n}^\top \bs v^\alpha = 1$. Define $\bs v:=[N^1 (\bs v^1)^\top, \dots, N^m (\bs v^m)^\top ]^\top $, where $N^\alpha$ is the total number of individuals belonging to class $\alpha$, for $\alpha \in \until{m}$. We call an equilibrium point $(\bs p^*, \bs x^*)$, an endemic equilibrium point, if at equilibrium the disease does not die out, i.e., $\bs p^* \neq 0$, otherwise, we call it a disease-free equilibrium point. Let $F^*:=F(\bs x^*)=F(\bs v)$ and $L^*:=L(\bs x^*)=L(\bs v)$. It can be verified that $F^*$ admits the splitting, $F^*= I-M$, where $I$ is the identity matrix of appropriate dimensions and $M$ is a Laplacian matrix. 

\begin{theorem}[\bit{Existence and Stability of Equilibria}] \label{thm:stability}
For the SIS model under multi-layer Markovian mobility~\eqref{eq_Model} with Assumption~\ref{Assumption:StrongConnectivity}, the following statements hold 
\begin{enumerate}
    \item if $\bs p(0) \in [0,1]^{nm}$, then $\bs p(t) \in [0,1]^{nm}$ for all $t>0$. Also, if $\bs p(0) > \bs 0_{nm}$, then $\bs p(t) \gg \bs 0_{nm}$ for all $t>0$;
    \item the model admits a disease-free equilibrium at $(\bs p^*, \bs x^*)= (\bs 0_{nm}, \bs v)$; 
    \item the model admits an endemic equilibrium at $(\bs p^*, \bs x^*) = (\bar{\bs p}, \bs v)$, $\bar{\bs p} \gg \bs 0$,  if and only if $\mu (BF^*-D-L^*) > 0$; 
    \item the disease-free equilibrium is globally asymptotically stable if and only if $\mu (BF^*-D-L^*) \leq 0$ and is unstable otherwise;
    \item the endemic equilibrium is almost globally asymptotically stable if $\mu (BF^*-D-L^*) > 0$ with region of attraction $\bs p(0) \in [0,1]^{nm}$ such that $\bs p(0) \neq \bs 0_{nm}$.
    \end{enumerate}

\end{theorem}
\medskip

\begin{proof}
The first part of statement (i) follows from the fact that $\dot{\bs{p}}$ is either  tangent or directed inside of the region $[0,1]^{nm}$ at its boundary which are surfaces with $p^\alpha_i =0$ or $1$ . This can be seen from \eqref{eq_SIS_pi}. For the second part of (i), we rewrite \eqref{eq_p} as:
\begin{equation*}
    \dot{\bs{p}} = ((I-P)BF(\bs x)+A(\bs{x}))\bs{p} - E(t) \bs{p}
\end{equation*}
where $L(\bs x)=C(\bs x)-A(\bs x)$ with $C(\bs x)$ composed of the diagonal terms of $L(\bs x)$, $A(\bs x)$ is the non-negative matrix corresponding to the off-diagonal terms, and $E(t)=C(\bs x(t))+D$ is a diagonal matrix. Now, consider a variable change $\bs y(t) := e^{\int_{0}^{t}E(\tau) d\tau}\bs p(t)$. Differentiating $\bs y (t)$ and using above gives:
\begin{align*}
    \dot{\bs{y}} & = e^{\int_{0}^{t}E(\tau) d\tau}((I-P)BF(\bs x)+A(\bs{x}))e^{\int_{0}^{t}-E(\tau) d\tau}e^{\int_{0}^{t}E(\tau) d\tau}\bs{p} \\
    &=e^{\int_{0}^{t}E(\tau) d\tau}((I-P)BF(\bs x)+A(\bs{x}))e^{\int_{0}^{t}-E(\tau) d\tau} \bs y
\end{align*}
Now, the coefficient matrix of $\bs y$ above is always non-negative and strongly connected. The rest of the proof is same as in \cite[Theorem 4.2 (i)]{meiBullo2017ReviewPaper_DeterministicEpidemicNetworks}.

The second statement follows by inspection.\\
The proof of the third statement is presented in Appendix \ref{Appendix: existence of non-trivial eqb}. 

\noindent\textbf{Stability of disease-free equilibria:}
To prove the fourth statement, we begin by establishing sufficient conditions for instability. The linearization of \eqref{eq_Model} at $(\bs p, \bs x) = (\bs 0, \bs v)$ is
\begin{equation} \label{eq_px linear}
    \begin{bmatrix}
     \dot{\bs p} \\
     \dot{\bs x}
     \end{bmatrix} = \begin{bmatrix}
     BF^*-D-L^* & 0 \\
     0 & Q^\top
     \end{bmatrix}\begin{bmatrix}
        \bs p \\
        \bs x
     \end{bmatrix} .
\end{equation}
Since the system matrix in~\eqref{eq_px linear} is block-diagonal, its eigenvalues are the eigenvalues of the block-diagonal sub-matrices. Further, since radial abscissa $\mu(Q^\top)$ is zero, a sufficient condition for instability of the disease-free equilibrium is that $\mu (BF^*-D-L^*) > 0$.

For the case of $\mu (BF^*-D-L^*) \leq 0$, we now show that the disease-free equilibrium is a globally asymptotically stable equilibrium.
It can be seen from the definitions of matrices $F^*$ and $L^*$, that under Assumption \ref{Assumption:StrongConnectivity}, $(BF^*-D-L^*)$ is an irreducible Metzler matrix. Together with $\mu (BF^*-D-L^*) \leq 0$, implies there exists a positive diagonal matrix $R$ such that 
\[
R(BF^*-D-L^*)+(BF^*-D-L^*)^ \top R = -K,
\]
where $K$ is a positive semi-definite matrix \cite[Proposition 1 (iv), Lemma A.1]{khanafer_Basar2016stabilityEpidemicDirectedGraph}.  Define $\Tilde{L} := L(\bs x)-L^*$, $\Tilde{F} := F(\bs x)-F^*$ and $r := \|R\|$, where $\|\cdot\|$  denotes the the induced two norm of the matrix. 

Since $\bs x(0) \gg 0$, under Assumption~\ref{Assumption:StrongConnectivity}, $x^\alpha_i(t)$ is lower bounded by some positive constant and hence, $\Tilde{L}$ and $\Tilde{F}$ are bounded and continuously differentiable.
Since $\bs x$ is bounded and exponentially converges to $\bs x^*$, it follows that $\|\Tilde{L}(\bs x)\|$ and $\|\Tilde{F}(\bs x)\|$ locally exponentially converge to $0$ and $\int_{0}^{t} \|\Tilde{L}\| d t$ and $\int_{0}^{t} \|\Tilde{F}\| d t$ are bounded for all $t>0$.

Consider the Lyapunov-like function $V(\bs p, t) = \bs p^\top R \bs p - 2 n m r \int_{0}^{t} (B\|\Tilde{F}\| + \|\Tilde{L}\|) d t$. 
It follows from the above arguments that $V$ is bounded. Therefore,
\begin{align}\label{Vdot_trivial}
     \dot{V} & = 2 \bs p^\top R \dot{\bs p} -2 n m r(B\|\Tilde{F}\| + \|\Tilde{L}\|) \nonumber \\
            & = \bs p^\top (R(BF^*-D-L^*)+(BF^*-D-L^*)^\top R) \bs p \nonumber \\
            & \quad + 2 \bs p^\top R (B\Tilde{F}-\Tilde{L})\bs p - 2\bs p^\top R P B F \bs p \nonumber \\
            & \quad-2 n m r(B\|\Tilde{F}\| + \|\Tilde{L}\|) \nonumber \\
            & = -\bs p^\top K \bs p + 2 \bs p^\top R (B\Tilde{F}-\Tilde{L})\bs p - 2\bs p^\top R P B F \bs p \nonumber  \\
            & \quad - 2 n m r(B\|\Tilde{F}\| + \|\Tilde{L}\|)  \nonumber \\
            & \leq -\bs p^\top K \bs p + 2 n m r(B\|\Tilde{F}\| + \|\Tilde{L}\|) \nonumber \\
            & \quad - 2 n m r(B\|\Tilde{F}\| + \|\Tilde{L}\|) - 2 \bs p^\top R P B F\bs p \nonumber  \\
            & \leq - 2 \bs p^\top R P B F \bs p \leq 0 .
\end{align}
Since all the signals and their derivatives are bounded, it follows that $\Ddot{V}(t)$ is bounded and hence $\dot{V}$ is uniformly continuous in $t$. Therefore from Barbalat's lemma and its application to Lyapunov-like functions ~\cite[Lemma 4.3, Chapter 4]{slotine1991applied} it follows that $\dot{V} \rightarrow 0$ as $t \rightarrow \infty$. Consequently, from \eqref{Vdot_trivial}, $\bs p^\top R P B F \bs p \rightarrow 0$. Since $R > 0$, $B > 0$, $F \geq 0$ with $F_{kk}>0$ and $ p_k \geq 0$, $\bs p(t) \rightarrow \bs 0$ as $t \rightarrow \infty$. This establishes global attractivity of the disease-free equilibrium point. We now establish its stability. 

We note that since, for $\bs x \gg 0$,  $(B\|\Tilde{F}\| + \|\Tilde{L}\|)$  is a real analytic function of $\bs x$, $\exists$ a region $\|\bs x - \bs x^*\|<\delta_1$ in which $(B\|\Tilde{F}\| + \|\Tilde{L}\|) \leq k_1\|\bs x - \bs x^*\|$ for some $k_1>0$. Also, since $\bs x - \bs x^*$ is globally exponentially stable, $\|\bs x(t) - \bs x^*\| \leq  k_2 e^{-\alpha t} \|\bs x(0) - \bs x^*\|$ for $k_2$, $\alpha >0$. Thus, if $\|\bs x(0) - \bs x^*\| < \frac{\delta_1}{k_2}$, then $(B\|\Tilde{F}\| + \|\Tilde{L}\|)\leq k_1 k_2 e^{-\alpha t}\|\bs x(0) - \bs x^*\|$. This implies $\int_{0}^{t} (B\|\Tilde{F}\| + \|\Tilde{L}\|) d t \leq \frac{k}{\alpha}\|\bs x(0) - \bs x^*\|$, where $k:=k_1 k_2$.  Now, since $\dot{V}(\bs p, t)\leq 0$, 
\begin{equation*} \label{eq_trivial stability}
\begin{split}
  V(\bs p(0), 0) &= \bs p(0)^\top R \bs p(0) \\
       &\geq V(\bs p(t), t) \\
       &\geq \bs p(t)^\top R \bs p(t) -2\frac{n m r k \|\bs x(0)-\bs x^*\|}{\alpha} \\
       &\geq \subscr{R}{min}\|\bs p(t)\|^2 - 2\frac{n m r k \|\bs x(0)-\bs x^*\|}{\alpha} ,
  \end{split}
\end{equation*}
  where $\subscr{R}{min} = \min_{i} (R_{i})$. Equivalently, 
\begin{equation*}
\begin{split}
  \|\bs p(t)\|^2 &\leq \frac{r}{\subscr{R}{min}} \|\bs p(0)\|^2 + 2\frac{n m r k\|\bs x(0)-\bs x^*\|}{\alpha \subscr{R}{min}}.
  \end{split}
\end{equation*}
It follows using stability of $\bs x$ dynamics, that for any $\epsilon >0$, there exists $\delta >0$ , such that $\| \bs x(0)-\bs x^*\|^2 + \| \bs p(0) \|^2 \leq \delta ^2 \Rightarrow \| \bs p(t)\|^2 + \| \bs x(t)-\bs x^*\|^2 \leq \epsilon ^2$. This establishes stability. Together, global attractivity and stability prove the fourth statement. 

\noindent\textbf{Stability of endemic equilibria:}
Finally, we prove the fifth statement. To this end, we first establish an intermediate result. 
\begin{lemma} \label{Lemma:p_i tends to 0 implies p tends to 0}
For the dynamics~\eqref{eq_p}, if $p^\alpha_{i}(t) \rightarrow 0$ as $t \rightarrow \infty$, for some $i \in \until{n}$ and $\alpha \in \until{m}$, then $\bs p(t) \rightarrow \bs 0$ as $t \to \infty$.\\
\end{lemma} 

\begin{proof}
It can be easily seen from \eqref{eq_SIS_pi} that $\Ddot{p}^\alpha_{i}$ is bounded and hence $\dot{p}^\alpha_{i}$ is uniformly continuous in $t$. Now if $p^\alpha_{i}(t) \rightarrow 0$ as $t \rightarrow \infty$, it follows from Barbalat's lemma \cite[Lemma 4.2]{slotine1991applied} that $\dot{p}^\alpha_{i} \rightarrow 0$. Therefore, from \eqref{eq_SIS_pi} and the fact that $- l^\alpha_{i j}(\bs x) \geq 0$ and $p^\alpha_{i} \geq 0$, it follows that $p^\alpha_{j}(t) \rightarrow 0$ for all $j$ such that $- l^\alpha_{i j} (\bs x) \neq 0$. Using Assumption~\ref{Assumption:StrongConnectivity} and applying the above argument for each class at each node implies
$\bs p(t) \rightarrow \bs 0$. 
\medskip
\end{proof}

Define $\Tilde{\bs p} := \bs p-\bs p^*$, $P^* := \operatorname{diag}(\bs p^*)$ and $\Tilde{P} := \operatorname{diag}(\Tilde{\bs p})$. Then
\begin{equation*}
    \begin{split}
        \dot{\Tilde{\bs p}} & =  (BF-D-L- P BF) \bs p \\
                        & =  (BF^*-D-L^*- P^* BF^*) \bs p^* \\
                        & \quad + (BF^*-D-L^*- P^* BF^*) \Tilde{\bs p} \\
                        & \quad + (B\Tilde{F}- \Tilde{L}) \bs p - PB\tilde{F} \bs p - \Tilde{P}BF^* \bs p \\
                        & = ((I- P^*)BF^*-D-L^*) \Tilde{\bs p} + ((I-P)B\tilde{F} - \Tilde{L}) \bs p \\
                        &\quad - \Tilde{P}BF^* \bs p .
    \end{split}
\end{equation*}
where we have used $(BF^*-D-L^*- P^* B F^*) \bs p^* = \bs 0$, as ($\bs p^*$, $\bs x^*$) is an equilibrium point.

Note that $(BF^*-D-L^*- P^* BF^*)=((I-P^*)BF^*-D-L^*)$ is an irreducible Metzler matrix and $\bs p^* \gg 0$ is its positive eigenvector associated with eigenvalue at zero. Therefore, the Perron-Frobenius theorem for irreducible Metzler matrices \cite{Bullo-book_Networks} implies 
$\mu ((I- P^*)BF^*-D-L^*) = 0$. 
Also, this means there exists a positive-diagonal matrix $R_2$ and a positive semi-definite matrix $K_2$ such that

\begin{align*}
   &R_{2}((I-P^*)BF^*-D-L^*) \\
   & \quad +((I-P^*)BF^*-D-L^*)^\top R_{2} = -K_2 .
\end{align*}

Similar to the proof of the fourth statement, take $V_{2}(\tilde{\bs p}, t) = \Tilde{\bs p}^\top R_{2} \Tilde{\bs p} - 2n m r_{2} \int_{0}^{t} (B\|\Tilde{F}\| + \|\Tilde{L}\|)  d t$, where $r_{2} := \|R_{2}\|$.
Then,
\begin{equation*} \label{Vdot_non-trivial}
\begin{split}
           \dot{V_{2}} & = 2 \tilde{\bs p}^\top R_{2} \dot{\tilde{\bs p}} -2n m r_{2} (B\|\Tilde{F}\| + \|\Tilde{L}\|) \\
            & = \tilde{\bs p}^\top (R_{2}((I-P^*)BF^*-D-L^*)\\
            &\quad +((I-P^*)BF^*-D-L^*)^\top R_{2}) \tilde{\bs p} \\
            & \quad + 2 \tilde{\bs p}^\top R_{2}((I-P)B\tilde{F}-\Tilde{L})\bs p - 2 \tilde{\bs p}^\top R_{2} \tilde{P} B F^*\bs p \\
            & \quad -2n m r_{2}(B\|\Tilde{F}\| + \|\Tilde{L}\|)\\
            & = -\tilde{\bs p}^\top K_{2} \tilde{\bs p} + 2 \tilde{\bs p}^\top R_{2}((I-P)B\tilde{F}-\Tilde{L})\bs p \\
            & \quad - 2 \tilde{\bs p}^\top R_{2} \tilde{P} B F^*\bs p -2n m r_{2}(B\|\Tilde{F}\| + \|\Tilde{L}\|)\\
            & \leq -\tilde{\bs p}^\top K_{2} \tilde{\bs p}  + 2 n m r_{2} (B\|\Tilde{F}\| + \|\Tilde{L}\|) \\
            & \quad- 2 n m r_{2} (B\|\Tilde{F}\| + \|\Tilde{L}\|)
            - 2 \tilde{\bs p}^\top R_{2} \tilde{P} BF^* \bs p\\
            & \leq - 2 \tilde{\bs p}^\top R_{2} \tilde{P} B F^* \bs p \\
            & \leq -2\displaystyle\sum_{k=1}^{n m} (R_2)_k\beta_k F^*_{k k}\tilde{p}_{k}^2 p_k \leq 0 .
\end{split}
\end{equation*}
The last inequality above follows from the fact that $R_2>0$, $B>0$ and $F\geq 0$ with diagonal terms $F_{k k}>0$. It can be easily shown that $\Ddot{V}_{2}$ is bounded implying $\dot{V}_{2}$ is uniformly continuous. Applying Barbalat's lemma \cite[Lemma 4.2]{slotine1991applied} gives $\dot{V}_{2} \rightarrow 0$ as $t \rightarrow \infty$. This implies that $\tilde{p}_{k} p_{k} \rightarrow 0$. Using Lemma \ref{Lemma:p_i tends to 0 implies p tends to 0}, and the fact that $\bs p= \bs 0$ is an unstable equilibrium for $\mu (BF^*-D-L^*) > 0$, we have $\tilde{\bs p} \rightarrow \bs 0$ as long as $\bs p(0) \neq \bs 0$.
Stability can be established similarly to the disease-free equilibrium case. This concludes the proof of the theorem. 
\end{proof}

\begin{corollary}[\bit{Stability of disease-free equilibria}] \label{cor:dis-free}
For the SIS epidemic model under Multi-layer Markovian mobility~\eqref{eq_Model} with Assumption~\ref{Assumption:StrongConnectivity} and the disease-free equilibrium $(\bs p^*, \bs x^*)= (\bs 0, \bs v)$ the following statements hold
\begin{enumerate}
    \item a necessary condition for stability is that for each $i \in \until{n}$, $\exists \alpha \in \until{m}$ such that $\delta_{i} > \beta_{i} - \nu^\alpha_{i}$; 
    \item  a necessary condition for stability is that there exists some $i \in \until{n}$ such that $\delta_i \geq \beta_i$; 
    \item a sufficient condition for stability is $\delta_{i} \geq \beta_{i}$, for each $i \in \until{n}$; 
    \item a sufficient condition for stability is 
    \[
    \frac{\lambda_{2}}{\Big(1+\sqrt{1+\frac{\lambda_{2}}{\sum_{i} w_{i}\big(\delta_{i}-\beta_{i}-s\big)}}\Big)^2 nm + 1} + s \geq 0,
    \]
    where 
    $\bs w$ is a positive left eigenvector of $(BM+L^*)$ such that $\bs w^\top(BM+ L^*) = 0$ with $\max_{i} w_{i} = 1$, $s = \min_{i} (\delta_{i}-\beta_{i})$, $W = \operatorname{diag} (\bs w)$, and  $\lambda_{2}$ is the second smallest eigenvalue of $\frac{1}{2}\big(W(BM+ L^*) +(BM+ L^*)^\top W\big)$. 
\end{enumerate}
\end{corollary}
\begin{proof}
We begin by proving the first two statements. First, we note that $(L^\alpha)^*_{ii} = \nu^\alpha_i$. This can be verified by evaluating $L^*=L(\bs v)$ and utilising the fact that $Q^\top \bs v = \bs 0$. The necessary and sufficient condition for the stability of disease-free equilibrium is $\mu (BF^*-D-L^*) \leq 0$. Note that $BF^*-D-L^*$ is an irreducible Metzler matrix. Perron-Frobenius theorem for irreducible Metzler matrices implies that there exists a real eigenvalue equal to $\mu$ with positive eigenvector, i.e.,
$(BF^*-D-L^*)\bs y = \mu \bs y $, where $\bs y \gg \bs 0 $. Rename components of $\bs y$ as $y_{(n\alpha+i)}=y^\alpha_i$ to write $\bs y = [(\bs y^1)^\top,\dots,(\bs y^m)^\top]^\top$. 

Let for each $i \in \until{n}$, $y^{k_i}_i=\min \{y^1_i,\dots,y^m_i\}$. Since $\mu \leq 0$, written component-wise for $(nk_i+i)$-th component

\begin{align*}
        & \sum_\alpha \beta_{i}f^{*\alpha}_i y^\alpha_i - (\delta_{i} +\nu^{k_i}_i)y^{k_i}_{i} - \displaystyle\sum_{j\neq i} l^{*k_i}_{i j}y^{k_i}_j \leq 0 \nonumber  \\
        & \Rightarrow \sum_\alpha \beta_{i}f^{*\alpha}_i y^{k_i}_i + \sum_\alpha \beta_{i}f^{*\alpha}_i (y^\alpha_i-y^{k_i}_i)- (\delta_{i} +\nu^{k_i}_i)y^{k_i}_i\nonumber\\
        & - \displaystyle\sum_{j\neq i} l^{*k_i}_{i j}y^{k_i}_j \leq 0 \nonumber\\
        & \Rightarrow (\beta_{i} - \delta_{i}-\nu^{k_i}_i)y^{k_i}_{i} \leq -\sum_\alpha \beta_{i}f^{*\alpha}_i (y^\alpha_i-y^{k_i}_i) + \displaystyle\sum_{j\neq i} l^{*k_i}_{i j}y^{k_i}_j  \nonumber \\
        & \Rightarrow (\beta_{i} - \delta_{i}-\nu^k_i)y^{k_i}_{i} < 0 \\
        & \Rightarrow \beta_{i} - \delta_{i}-\nu^{k_i}_i < 0 .
    \end{align*}
Here we have used facts: $\sum_\alpha f^{*\alpha}_i = 1$, $f^{*\alpha}_i>0 $, $ l^{*k_i}_{ij}\leq 0$ and that there exists $j \in \until{n}$ such that $ l^{*k_i}_{ij}<0$. This proves the statement (i). 

Let $y^{k_i}_i$ be $\min\{y^1_1,\dots,y^m_n\}$. Similar to the proof of the first statement

\begin{align*}
        & \sum_\alpha \beta_{i}f^{*\alpha}_i y^\alpha_i - (\delta_{i} +\nu^{k_i}_{i})y^{k_i}_{i} - \displaystyle\sum_{j\neq i} l^{*k_i}_{i j}y^{k_i}_j \leq 0 \nonumber  \\
        & \Rightarrow \sum_\alpha \beta_{i}f^{*\alpha}_i y^{k_i}_i + \sum_\alpha \beta_{i}f^{*\alpha}_i (y^\alpha_i-y^{k_i}_i)- (\delta_{i} +\nu^{k_i}_{i})y^{k_i}_i\nonumber\\
        & - \displaystyle\sum_{j\neq i} l^{*k_i}_{i j}y^{k_i}_i - \displaystyle\sum_{j\neq i} l^{*k_i}_{i j}(y^{k_i}_j-y^{k_i}_i) \leq 0 \nonumber\\
        & \Rightarrow (\beta_{i} - \delta_{i})y^{k_i}_{i} \leq -\sum_\alpha \beta_{i}f^{*\alpha}_i (y^\alpha_i-y^{k_i}_i) + \displaystyle\sum_{j\neq i} l^{*k_i}_{i j}(y^{k_i}_j -y^{k_i}_i) \nonumber \\
        & \Rightarrow (\beta_{i} - \delta_{i})y^{k_i}_{i} \leq 0 \\
        & \Rightarrow \beta_{i} - \delta_{i} \leq 0 .
    \end{align*}

Here we have used an additional fact:  $\nu^{k_i}_i+ \displaystyle\sum_{j\neq i} l^{*k_i}_{i j}=0$. This proves statement (ii).\\

Let $F^*=I-M$ where $M$ is a Laplacian matrix which can be seen from the definition of $F$. Now $BF^*-D-L^* = B-D-(BM+L^*)$. Since $(BM+L^*)$ is an irreducible Laplacian matrix, if $\delta_i \geq \beta_i$, for each $i \in \until{n}$, from Gershgorin disks theorem \cite{Bullo-book_Networks}, $\mu \leq 0$, which proves the third statement.

For the last statement, we use an eigenvalue bound for perturbed irreducible Laplacian matrix of a digraph~ \cite[Theorem 6]{wu2005bounds}, stated below:

Let $H = A + \Delta$, where $A$ is an $n\times n$ irreducible Laplacian matrix and $\Delta \neq 0$ is a non-negative diagonal matrix, then  
\begin{equation*}
\begin{split}
  \mathrm{Re}(\lambda(H)) \geq \frac{\lambda_{2}}{\Big(1+\sqrt{1+\frac{\lambda_{2}}{\sum_{i} w_{i}\Delta_i}}\Big)^2 n + 1} > 0,
  \end{split}
\end{equation*}
where, $\bs w$ is a positive left eigenvector of $A$ such that $\bs w^\top A = 0$ with $\max_{i} w_{i} = 1$, $W = \operatorname{diag} (\bs w)$, and  $\lambda_{2}$ is the second smallest eigenvalue of $\frac{1}{2}(W A + A^\top W)$.\\
Now, in our case necessary and sufficient condition for stability of disease-free equilibrium is:
\begin{equation*}
\begin{split}
  \mathrm{Re}(\lambda(BM+L^*+D-B)) & = \mathrm{Re}(\lambda(BM+L^*+\Delta + sI)) \\
  & = \mathrm{Re}(\lambda(BM+L^*+\Delta)) + s \\
  & \geq 0,
  \end{split}
\end{equation*}
where, $s = \min_{i} (\delta_{i}-\beta_{i})$ and $\Delta=D-B-sI$. Applying the eigenvalue bound with $H=BM+L^*+\Delta$ gives the sufficient condition (iv). 
\end{proof}

\medskip
\begin{remark}
For given graphs and the associated mobility transition rates in dynamics~\eqref{eq_Model}, let $s = \operatorname{min}_{i} (\delta_{i}-\beta_{i})$ and $i^*= \argmin_{i} (\delta_{i}-\beta_{i})$. Then, there exist $\delta_i$'s, $i\neq i^*$, that satisfy statement (iv) of Corollary \ref{cor:dis-free} if $s > \subscr{s}{lower}$, where 
\[
\subscr{s}{lower}=-\frac{\lambda_2}{4mn+1}.
\] \oprocend
\end{remark}

\begin{remark}(\bit{Influence of mobility on stability of disease-free equilibrium.})
The statement (iv) of Corollary~\ref{cor:dis-free} characterizes the influence of mobility on the stability of disease-free equilibria. In particular, $\lambda_2$ is a measure of ``intensity" of mobility and $s$ is a measure of largest deficit in the recovery rate compared with infection rate among nodes. The sufficient condition in statement (iv) states explicitly how mobility can allow for stability of disease-free equilibrium even under deficit in recovery rate at some nodes. \oprocend
\end{remark}

\section{Numerical Illustrations} \label{Sec: numerical studies}

We start with numerical simulation of epidemic model with multi-layer mobility in which we treat epidemic spread as well as mobility as stochastic processes. The fraction of infected populations for different cases are shown in Fig.~\ref{fig:Stochastic}. The corresponding simulations of the deterministic model as per Proposition \ref{prop:model} are also shown for comparison. We take two mobility network layers: a complete graph and a line graph with the mobility transition rates being equal among out going neighbors of a node for both the graphs. The two cases relate to the stable disease-free equilibrium and stable endemic equilibrium respectively.  If the curing rates, infection rates and the initial fraction of infected population are the same for all the nodes, mobility does not play any role. Therefore, we have chosen heterogeneous curing or infection rates to elucidate the influence of mobility.  Figure~\ref{fig:Stochastic}~(a) corresponds to the case $\delta_i \geq \beta_i$ for each $i$, whereas Fig.~\ref{fig:Stochastic}~(c) corresponds to the case $\delta_i< \beta_i$ for each $i$. The results support statements (iii) and (ii) of Corollary \ref{cor:dis-free} and lead to, respectively, the stable disease-free equilibrium and the stable endemic equilibrium.

\begingroup
\centering
\begin{figure}[ht!]
\centering
\subfigure[Stable disease-free equilibrium: Stochastic model]{\includegraphics[width=0.23\textwidth]{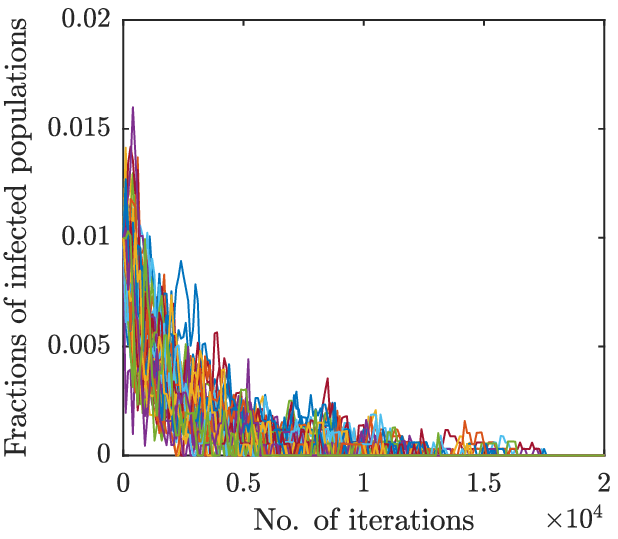}}\label{fig:Stochastic_a}
\subfigure[Stable disease-free equilibrium: Deterministic model]{\includegraphics[width=0.23\textwidth]{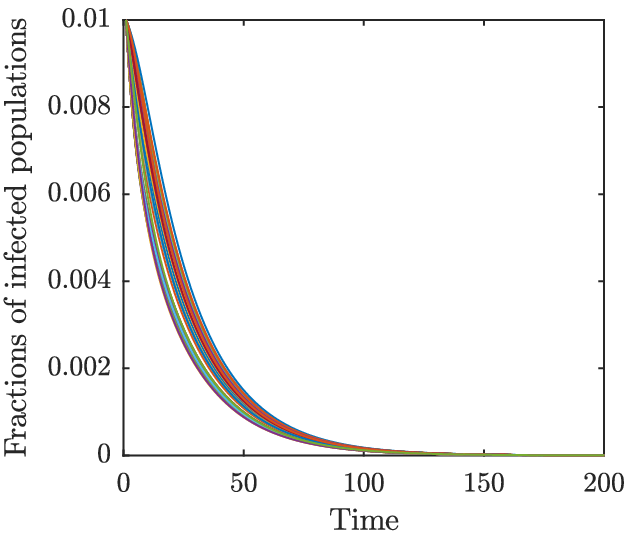}}\label{fig:Stochastic_b}
\subfigure[Stable endemic equilibrium: Stochastic model]{\includegraphics[width=0.23\textwidth]{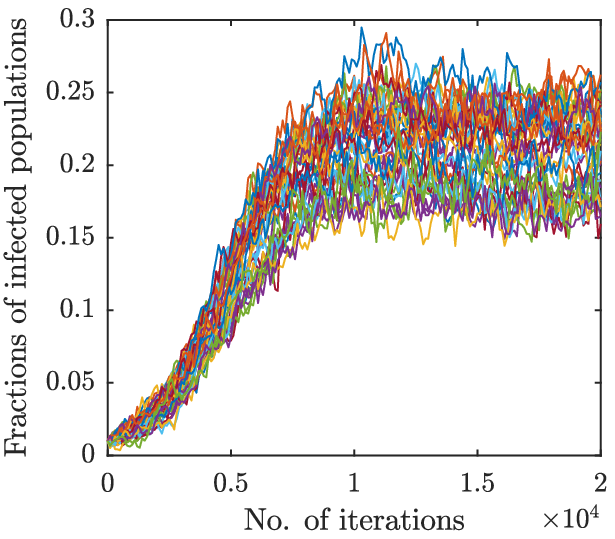}}\label{fig:Stochastic_c}
\subfigure[Stable endemic equilibrium: Deterministic model]{\includegraphics[width=0.23\textwidth]{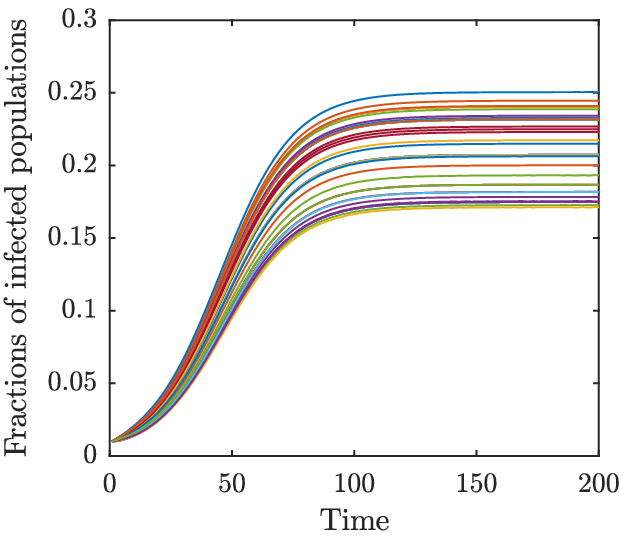}}\label{fig:Stochastic_d}
\caption{Stochastic simulation of epidemic spread under mobility. Complete-Line graphs, $n=20$, $\nu(i) = 0.2$, $q_{i j}=\frac{\nu(i)}{D_{out}}$, $p_i(0)=0.01$. Each iteration in stochastic model corresponds to time-step $0.01$ sec.}
\label{fig:Stochastic}
\end{figure}
\endgroup
Once we have established the correctness of deterministic model predictions with the stochastic simulations, we study the simulations using only the deterministic model. We study the effect of multi-layer mobility over different pairs of mobility graph structures - line-line graph, line-ring graph and line-star graph. We choose different population size for the two mobility layers and take the mobility transition rates such as to keep the equilibrium distribution of population the same for both the layers across all pairs (taken as uniform equilibrium distribution) by using instantaneous transition rates from Metropolis-Hastings algorithm \cite{Hastings_MetroplisHastingsMC}. This shows the effect of different mobility graph structure on epidemic spread while the equilibrium population distribution remains the same. Fig.~\ref{fig:Deterministic_SameMobilityEqbDist} shows the fractions of infected population trajectories for $10$ nodes connected with different pairs of graph structures. The values of equilibrium fractions are affected by the presence of mobility and are different for different graph structures.

\begingroup
\centering
\begin{figure}[ht!]
\centering
\subfigure[Line-Line graphs; graph 1]{\includegraphics[width=0.23\textwidth]{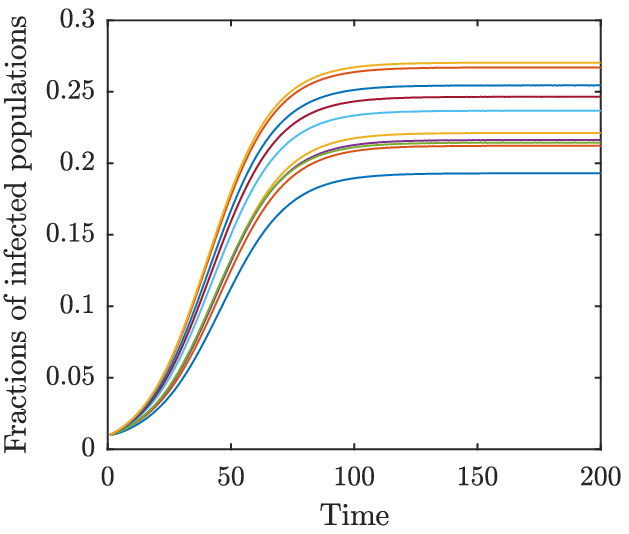}}\label{fig:Samedist_L-L_1}
\subfigure[Line-Line graphs; graph 2]{\includegraphics[width=0.23\textwidth]{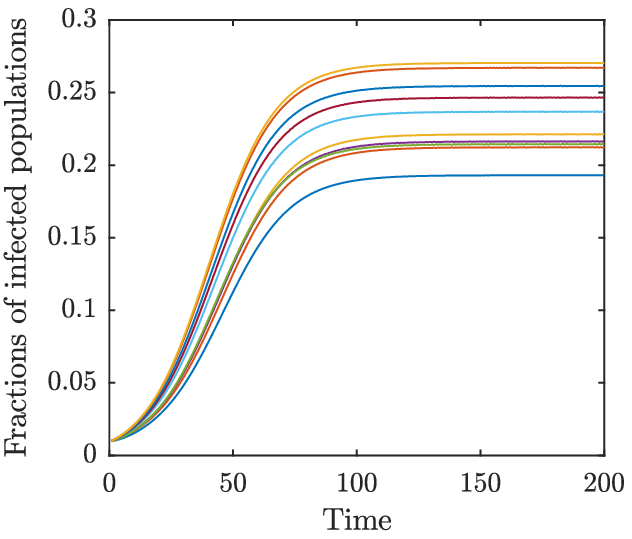}}\label{fig:Samedist_L-L_2}
\subfigure[Line-Ring graphs; graph 1]{\includegraphics[width=0.23\textwidth]{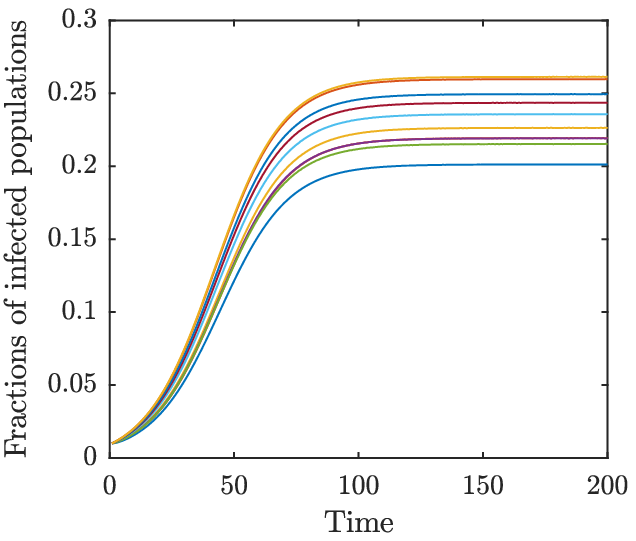}}\label{fig:Samedist_L-R_1}
\subfigure[Line-Ring graphs; graph 2]{\includegraphics[width=0.23\textwidth]{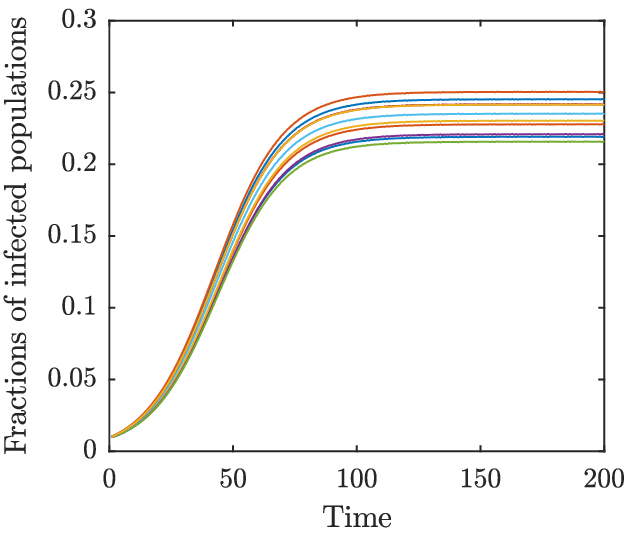}}\label{fig:Samedist_L-R_2}
\subfigure[Line-Star graphs; graph 1]{\includegraphics[width=0.23\textwidth]{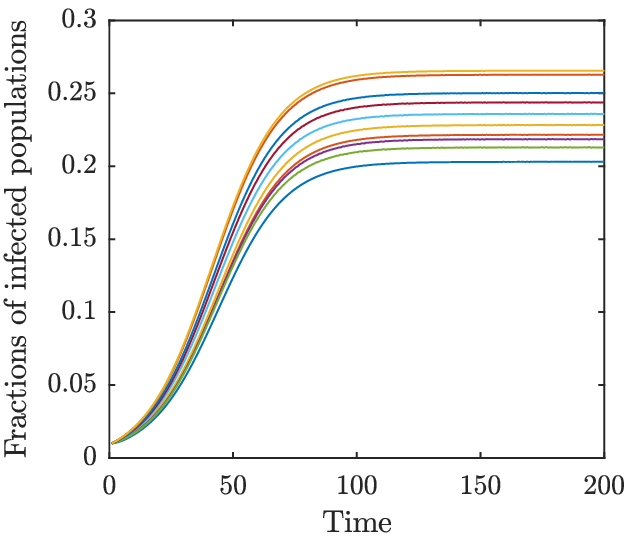}}\label{fig:Samedist_L-S_1}
\subfigure[Line-Star graphs; graph 2]{\includegraphics[width=0.23\textwidth]{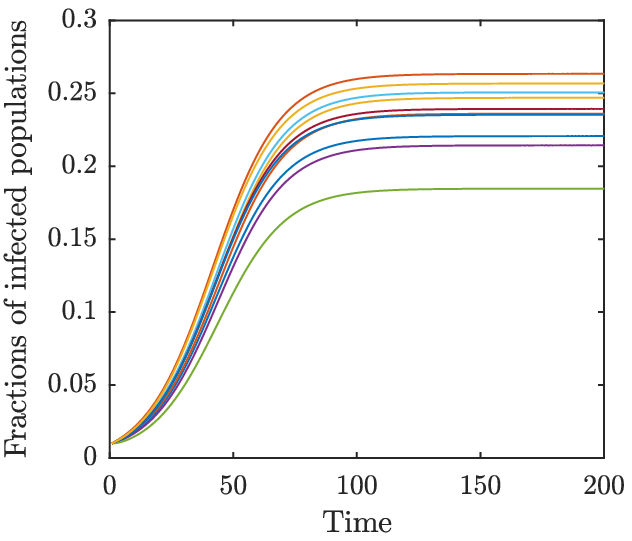}}\label{fig:Samedist_L-S_2}
\caption{Simulation of deterministic model of epidemic spread under 2 layer mobility, over different graph structure with stable endemic equilibrium. $n=10$, $p_i(0)=0.01$.}
\label{fig:Deterministic_SameMobilityEqbDist}
\end{figure}
\endgroup

Next, we verify the statement (iv) of Corollary \ref{cor:dis-free}, for a single layer mobility model, where one can have some curing rates $\delta_i$ less than the infection rates $\beta_i$ but still have stable disease-free equilibrium. We take a complete graph of $n=20$ nodes with given mobility transition rates which give us $\bs w$, $L^*$ and $\lambda_2$. We take a given set of values of $\beta_i$. Next, we compute $\subscr{s}{lower} = -\frac{\lambda_2}{4nm+1}$ and take $0.8$ times of this value as $s$ in order to compute $\delta_i$'s that satisfy statement (iv) of Corollary \ref{cor:dis-free}. For our case the values are: $\beta_i=0.3$, $\lambda_2=0.2105$, $\subscr{s}{lower}=-0.0026$, $s = 0.8~ \subscr{s}{lower}=-0.0021$, $\delta_1=\delta_n=\beta_i+s$ and the rest $\delta_i$ computed to satisfy the condition which gives $\delta_1 = \delta_n = 0.2979$ and $\delta_i = 0.3198$ for $i \in \{2,\dots,n-1\}$. Fig.~\ref{fig:Lambda2 sufficient cond Complete graph} shows the trajectories of infected fraction populations. As can be seen the trajectories converge to the disease-free equilibrium.

\begin{figure}[ht!]
    \centering
    \includegraphics[width=0.9\linewidth]{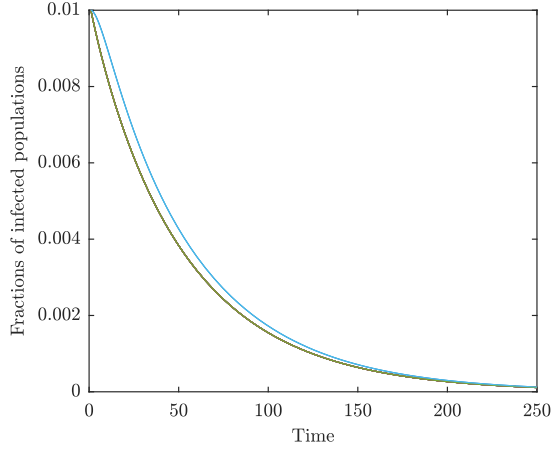}
    \caption{Stable disease-free equilibrium with curing rates computed as per the $\lambda_2$ sufficient condition (statement (iv), Corollary \ref{cor:dis-free}) for stability of disease-free equilibrium. Graph: Complete, $n=20$, $p_i(0)=0.01$.}
    \label{fig:Lambda2 sufficient cond Complete graph}
\end{figure}

\section{Conclusions} \label{Sec: conclusions}

We derived a continuous-time model for epidemic propagation under Markovian mobility across multi-layer network of patches. The epidemic spread within each node has been modeled as SIS population model with individuals traveling across the nodes with different mobility patterns modeled as layers of a multi-layer network. The derived model has been analysed to establish the existence and stability of disease-free equilibrium and an endemic equilibrium under different conditions.
Some necessary and some sufficient conditions for stability of disease-free equilibrium have been established. We also provided numerical studies to support our results and elucidated the effect of mobility on epidemic propagation.

\appendix
\subsection{Proof of Theorem 1 (iii): Existence of an endemic equilibrium} \label{Appendix: existence of non-trivial eqb}
Here we assume that there exists atleast one node with positive recovery rate, i.e., $\delta_i > 0$ for atleast one $i$. The case with no recovery at all nodes is trivial and leads to $\bs p^* = \bs 1$.

We first state some properties of M-matrices, which we will use in the proof.

\begin{theorem}[\bit{Properties of M-matrix, \cite{berman1994nonnegative}}]\label{M-matrix properties}
For a real Z-matrix (i.e., a matrix with all off-diagonal terms non-positive) $A\in \real^{n\times n}$, the following statements are equivalent to $A$ being a non-singular M-matrix
\begin{enumerate}
    \item \bit{Stability}: real part of each eigenvalue of $A$ is positive;
    \item \bit{Inverse positivity}: $A^{-1}\geq0$ (for irreducible $A$, $A^{-1}>0$);
    \item \bit{Regular splitting}: $A$ has a convergent regular splitting, i.e., $A$ has a representation $A=M-N$, where $M^{-1}\geq 0$, $N\geq0$ (called regular splitting), with $M^{-1}N$ convergent, i.e., $\rho(M^{-1}N)<1$;
    \item \bit{Convergent regular splitting}: every regular splitting of $A$ is convergent. Further, for a singular M-matrix (i.e. singular Z-matrix with real part of eigenvalues non-negative) regular splitting of $A$ gives $\rho(M^{-1}N)=1$;
    \item \bit{Semi-positivity}: there exists $\bs x \gg \bs 0$ such that $A\bs x \gg \bs 0$;
     \item \bit{Modified semi-positivity}: there exists $\bs x \gg \bs 0$ such that $\bs y = A\bs x > \bs 0$ and matrix $\hat A$ defined by
     \begin{equation*}
         \hat A_{ij}=\begin{cases}
                      1  & \text{if } A_{ij}\neq 0 \text{ or } y_i \neq 0,\\
                      0  & \text{otherwise}
                      \end{cases}
     \end{equation*}
     is irreducible.
    
\end{enumerate}

\end{theorem}
\medskip

A consequence of Theorem \ref{M-matrix properties} (vi) is that an irreducible Laplacian matrix perturbed with a non-negative diagonal matrix with atleast one positive element is a non-singular M-matrix (take $\bs x = \bs 1 \gg \bs 0$). This implies that block diagonal submatrices of the matrix $L^*+D$ are all non-singular M-matrices (since $\delta_i \geq 0$ with strict inequality for atleast one $i$) and hence $L^*+D$ is a non-singular M-matrix. Similar arguments imply $B^{-1}(L^*+D)$ is a non-singular M-matrix.

We show below that in the case of $\mu (BF^*-D-L^*) > 0$ there exists an endemic equilibrium $\bs p^* \gg 0$, i.e.,
\begin{equation} 
    \dot{\bs p} |_{\bs p =\bs  p^*} = (BF^*-D-L^*- P^* BF^*) \bs p^* = 0 .
\end{equation}

We use Brouwer's fixed point theorem, similar to the derivation in \cite{fall2007epidemiological}. We split the non-negative matrix $F^*$ as $F^* = I-M$, where $M$ is a Laplacian matrix. Rearranging the terms and writing the above as an equation in $\bs p$ to be satisfied at $\bs p^*$ leads to

\begin{equation} \label{eqAppendixeqbM}
\begin{split}
    (L^*+D)((L^*+D)^{-1} B - I)\bs p & =  (PB +(I-P)BM) \bs p \\
    & = B(P + (I-P)M) \bs p .
\end{split}
\end{equation}

Define $A := (L^*+D)^{-1} B$. Since $A^{-1} = B^{-1} (L^*+D)$ is a non-singular M-matrix, its inverse $A$ is non-negative \cite{berman1994nonnegative}. Rearranging  \eqref{eqAppendixeqbM} leads to 
\begin{equation}
   \bs p = H (\bs p) = (I + A(P+(I-P)M))^{-1}A \bs p .
\end{equation}

Now we show that $H(\bs p)$ as defined above is a monotonic function in the sense that $\bs p_{2} \geq \bs p_{1}$ implies $H(\bs p_{2}) \geq H(\bs p_{1})$. Define $\tilde{\bs p} := \bs p_2 - \bs p_1$ and $\tilde{P} := \operatorname{diag}(\tilde{\bs p})$. Then,

\begin{equation}
 \begin{split} \label{eqHMultiplex}
       & H(\bs p_{2}) - H(\bs p_{1}) \\
        &=
          \left(A^{-1}+P_{2} + (I-P_{2})M\right)^{-1}\bs p_{2} \\
         &\quad - \left(A^{-1}+P_{1}+(I-P_{1})M\right)^{-1}\bs p_{1} \\
         & = \left(A^{-1}+P_{2} + (I-P_{2})M\right)^{-1}\Big(\bs p_{2} -\\
         &\quad \left(A^{-1}+P_{2}+ (I-P_{2})M\right)\left(A^{-1}+P_{1}+(I-P_{1})M\right)^{-1}\bs p_{1}\Big)\\
         & = \left(A^{-1}+P_{2} + (I-P_{2})M\right)^{-1} \Big(\tilde{\bs p} \\
         &\quad- \tilde{P}(I-M)\left(A^{-1}+P_{1}+(I-P_{1})M\right)^{-1}\bs p_{1}\Big)\\
         & = (A^{-1}+P_{2} + (I-P_{2})M)^{-1} \Big(I \\
         &\quad - \operatorname{diag}\left((I-M)(A^{-1}+P_{1}+(I-P_{1})M)^{-1}\bs p_{1}\right)\Big)\tilde{\bs p}\\
   \end{split}
\end{equation}

Since $(A^{-1}+P_{2} + (I-P_{2})M) = B^{-1}(L^*+D) + P_{2} +(I-P_{2})M$ is a non-singular M-matrix (consider theorem \ref{M-matrix properties} (vi) with $\bs x =\bs 1 \gg \bs 0$), its inverse and hence the first term above is non-negative. The second term is shown to be non-negative as below

\begin{equation} \label{eqIAPMultiplex}
    \begin{split}
       & \Big(I - \operatorname{diag}\left((I-M)(A^{-1}+P_{1}+(I-P_{1})M)^{-1}\bs p_{1}\right)\Big) \\
        & = \Big(I - \operatorname{diag}\left((I-M)(A^{-1}+P_{1}+(I-P_{1})M)^{-1} P_{1} \bs 1\right)\Big) \\
        & = \operatorname{diag}\Big(\left(I - (I-M)(I + A P_{1} + A (I-P_{1})M )^{-1} A P_{1}\right) \bs 1\Big) \\
         & = \operatorname{diag}\Big(\big(I - M - (I-M)(I + A P_{1} + A (I-P_{1})M )^{-1} \\
         & \quad \left(A P_{1} + A (I-P_{1})M\right)\big) \bs 1\Big) \\
        & = \operatorname{diag}\Big((I-M)\big(I - (I + A P_{1} + A (I-P_{1})M )^{-1} \\
         & \quad \left(A P_{1} + A (I-P_{1})M \right)\big) \bs 1\Big) \\
         & = \operatorname{diag}\Big((I-M)\left(I + A P_{1} + A (I-P_{1})M \right)^{-1} \bs 1\Big) \\
         & = \operatorname{diag}\Big(F^* \big(A^{-1} + P_{1} + (I-P_{1})M \big)^{-1} A^{-1} \bs 1\Big) \\
         & \geq 0 ,
    \end{split}
\end{equation}
where we have used the identity
\begin{equation}
    (I + X)^{-1} = I - (I+X)^{-1}X ,
\end{equation}
and $M \bs 1 = \bs 0$ , as $M$ is a Laplacian matrix. The last inequality in \eqref{eqIAPMultiplex} holds as $A^{-1} \bs 1 = B^{-1} (L^* + D) \bs 1 = B^{-1} D \bs 1 \geq \bs 0$ and $(A^{-1} + P_{1} + (I-P_{1})M)^{-1} \geq \bs 0 $, since it is the inverse of an M-matrix. The last term in the last line of \eqref{eqHMultiplex} is $\tilde{\bs p} \geq \bs 0$. This implies that $H(\bs p)$ is a monotonic function. Also, argument similar to above can be used to show that $H(\bs p) \leq \bs 1$ for all $\bs 0 \leq \bs p \leq \bs 1$. Therefore, $H(\bs 1) \leq \bs 1$.\\

Applying the converse of Theorem \ref{M-matrix properties}~(iv), with Z-matrix as $(L^*+D)-BF^*$, where $(L^*+D)^{-1}\geq0$, $BF^*\geq0$ implies $\mu \left(BF^*-(D+L^*)\right) > 0$ if and only if $R_0= \rho(AF^*) = \rho(A(I-M)) > 1$. Now, $A$ is a block-diagonal matrix with block-diagonal terms as $A^{\alpha} = (L^{*\alpha}+D^{\alpha})^{-1}B^{\alpha}$, which are inverse of irreducible non-singular M-matrices and hence are positive. Using the expression for $F$ gives $AF^*= [(A^1)^\top \bar{F}^\top(\bs x^*),\dots,(A^m)^\top \bar{F}^\top(\bs x^*)]^\top$. Since $A^\alpha > 0$ and $\bar F^* \geq 0 $ with no zero column, $AF^*>0$ and hence irreducible.  Since $AF^*$ is an irreducible non-negative matrix, Perron-Frobenius theorem implies $\rho(AF^*)$ is a simple eigenvalue satisfying $AF^* \bs u = \rho (AF^*) \bs u = R_0 \bs u$ with $\bs u \gg \bs 0$. Using $F^*=I-M$ implies:

\begin{equation}
\begin{split}
    A\bs u &= R_0\bs u + AM\bs u\\
    &= (R_0 -1)\bs u + (I+AM)\bs u .
\end{split}
\end{equation}

Define $U :=\operatorname{diag}(\bs u)$ and $\gamma := \frac{R_0 -1}{R_0}$. Putting $\bs p = \epsilon \bs u$ , we show that $\exists$ $\epsilon _0$ such that  $\epsilon \in (0, \epsilon_0)$ implies  $H(\epsilon \bs u)\geq \epsilon \bs{u}$ as below:

\begin{equation}
    \begin{split}
        & H(\epsilon \bs u) - \epsilon \bs u \\
        & = \big(I+\epsilon AU + A(I-\epsilon U)M\big)^{-1} A\epsilon \bs u - \epsilon \bs u \\
        & = \epsilon \Big(\big(I+\epsilon AU + A(I-\epsilon U)M\big)^{-1}(R_0 -1) \bs u \\
        & \quad + \big(I+\epsilon AU + A(I-\epsilon U)M\big)^{-1}(I+AM) \bs u - \bs u\Big) \\
        &\equiv \epsilon K(\epsilon) .
    \end{split}
\end{equation}

Now we evaluate $ K(\epsilon)$ at $\epsilon = 0$ :

\begin{equation}
    \begin{split}
        &K(0) \\
        &= (I+AM)^{-1}(R_0 -1)\bs u + (I+AM)^{-1} (I+AM) \bs u - \bs u \\
        &= (I+AM)^{-1}(R_0 -1) \bs u \\
        &= (R_0 -1)(A^{-1}+M)^{-1}A^{-1} \bs u \\
        &= \frac{(R_0 -1)}{R_0}(A^{-1}+M)^{-1}F^* \bs u \\
        &= \gamma (A^{-1}+M)^{-1}F^* \bs u \\
        & \gg \bs 0 .
    \end{split}
\end{equation}

The last inequality follows as $\gamma$ and $\bs u$ are both positive, and $(A^{-1}+M)^{-1}F^* = (B^{-1}(L+D)+M)^{-1} F^*> 0$ as $B^{-1}(L+D)+M$ is an irreducible M-matrix and hence its inverse is positive and $F^* \geq 0$  with no zero column. 
Since $K(\epsilon)$ is a continuous function of $\epsilon$ , $\exists$ $\epsilon_0$ such that $\epsilon _0 > \epsilon >0$ implies $K(\epsilon) \gg \bs 0$ and therefore, $H(\epsilon \bs u)\geq \epsilon \bs{u}$.
Therefore there exists an $\epsilon > 0 $ such that $H(\epsilon \bs u) - \epsilon \bs u \geq \bs 0$ or equivalently, $H(\epsilon \bs u)\geq \epsilon \bs{u}$. Taking the closed compact set $J = [ \epsilon \bs u, \bs 1]$, $H(\bs p): J \rightarrow J$ is a continuous function of $\bs p$. Brouwer's fixed point theorem implies there exists a fixed point of $H$ in $J$. This proves the existence of an endemic equilibrium $\bs p^* \gg \bs 0$ when $\mu (BF^*-D-L^*) > 0$ or equivalently $R_0 >1$.

\balance
{\footnotesize
\bibliographystyle{ieeetr}
\bibliography{mybib}}

\end{document}